\documentclass[12pt]{amsart}
\usepackage{amsmath,amsfonts,amstext,amsthm,amssymb,pxfonts}
\usepackage[colorlinks,citecolor=red,hypertexnames=false]{hyperref}
\usepackage{amsrefs}
\usepackage[pdftex]{graphicx}

\theoremstyle{plain} 
\newtheorem{lemma}[equation]{Lemma} 
 
\newtheorem{theorem}[equation]{Theorem}

\theoremstyle{definition}
\newtheorem{definition}[equation]{Definition} 

\theoremstyle{remark}

\allowdisplaybreaks

\numberwithin{equation}{section}

\title{Uniqueness theorem for analytic functions and its application in denoising problem}

\author[ A.~Vagharshakyan ] {Ashot Vagharshakyan }
\thanks{}
\address{ Institute of Mathematics National Academy of 
Siences of Armenia, Bagramian 24-b, 375019 Yerevan, Armenia.}
\email{vagharshakyan@yahoo.com}

\begin{document}
\begin{abstract}
	In various applications the problem of separation of the original signal and the noise arises. 
        For example, in the identification problem for discrete linear and causal systems, 
        the original signal consists of the values of transfer function at some points in the unit disk. 
        In this paper we discuss the problem of choosing the points in the unite disk, 
        for which it is possible to remove the additive noise with probability one. 
        Since the transfer function is analytic in the unite disk, so this problem is related to the uniqueness theorems for analytic functions. 
        Here we give a new uniqueness result for bounded analytic functions and show its applications in the denoising problem.
\end{abstract}

\maketitle

\section{Introduction.}

In this paper we discuss a uniqueness theorem for some classes 
of bounded analytic functions, where the conclusion $f(z)\equiv 0$ follows 
from the condition that $f(z_n)$ goes to zero over some sequence of 
points $z_n,\quad n=1, 2, \dots$. Many results of S.J. Havinson 
are dedicated to this type results and they applications in approximation
theory. This type results one can find in 
\cite{MR808672}.

	The survey on mentioned above type uniqueness results for analytic functions one 
can find in \cite{MR1965913}. A new result of that type we give in this paper. Furder 
we apply that result in de-noising problem.

\section{Auxiliary results}
At first let us give here some auxiliary definitions and results.
\begin{definition}
Let $h(t),\quad 0<t,$ be a continuous, non-negative function. 
Let the family of arcs $\{S_k\}_{k=1}^{\infty}$ cover a given set $E$;
\begin{equation}
E\subseteq \bigcup_{k=1}^{\infty}S_k.
\end{equation}
Let us put
\begin{equation}
M_h(E)=\inf \sum_{k=1}^{\infty}h(|S_k|),
\end{equation}
where $|S|$ is the length of the arc $S$ and the minimum is taken over the family of all covers.
\end{definition}
\begin{definition}
One say that a function $f(t),\quad -\pi<t<\pi,$ belongs Bessov space $B_2^{\alpha}$ if
\begin{equation}
||f||_{\alpha}^2=\int_{-\pi}^{\pi}|f(x)|^2dx+
\int_{-\pi}^{\pi}\int_{-\pi}^{\pi}\frac{|f(x)-f(y)|^2}{|x-y|^{1+2\alpha}}dxdy<\infty,
\end{equation}
where $0<\alpha<1$.
\end{definition}
	For arbitrary function $g(t)\in L_1(-\pi,\pi)$ let us denote
\begin{equation}
g^*(t)=\sup_{0<\delta}\frac{1}{2\delta}\int_{t-\delta}^{t+\delta}|g(x)|dx,
\end{equation}
where $g(t)$ is assumed to be continued as a $2\pi$ periodic function on $(-\infty,+\infty)$.

	For a subset $E$ the quantity
\begin{equation}
C_{\alpha}(E)=\left(\inf_{\mu}\int_{-\pi}^{\pi}\int_{-\pi}^{\pi}\frac{d\sigma(x)d\sigma(y)}{|x-y|^{\alpha}}\right)^{-1},
\end{equation}
where the minimum is taken over the probability measures with support in $E$, is known as $\alpha-$ 
capacity of the subset $E$.

The following lemma is announced in the book \cite{MR0225986}, p.35.
\begin{lemma}
 Let $E$ be Borelian set and $C_{\alpha}(E)=0,$ where $0<\alpha<1$. Let $0\leq h(r), \quad r>0,$ 
be an increasing function and 
\begin{equation}
\int_0\frac{dh(r)}{r^{\alpha}}<\infty.
\end{equation}
Then
\begin{equation}
M_h(E)=0.
\end{equation}
\end{lemma}
\begin{lemma}
Let $f(x)\in B_2^{\alpha}$ and $0<\beta\leq\alpha$. 
Then there is a function $g(x)\in B_2^{\alpha-\beta}$ such that
\begin{equation}
\frac{1}{2\pi}\int_{-\pi}^{\pi}\frac{1-|z|^2}{|z-e^{it}|^2}f(t)dt=
\end{equation}
\begin{equation}
=\frac{1}{2\pi}\int_{-\pi}^{\pi}\frac{g(t)}{(z-e^{it})^{1-\beta}}dt+
\frac{1}{2\pi}\int_{-\pi}^{\pi}\frac{g(t)}{({\bar z}-e^{it})^{1-\beta}}dt,\quad |z|<1.
\end{equation}
\end{lemma}
\begin{proof}
Let 
\begin{equation}
f(x)\sim\sum_{k=-\infty}^{\infty}a_ke^{ikx}.
\end{equation}
be the Fourier series of the function $f(x)$. Since $f(x)\in B_2^{\alpha}$, so
\begin{equation}
\sum_{k=-\infty}^{\infty}|a_k|^2|k|^{2\alpha}<\infty.
\end{equation}
We can write
\begin{equation}
F(re^{ix})=\sum_{k=-\infty}^{\infty}a_kr^{|k|}e^{ikx}=
\end{equation}
\begin{equation}
=\frac{1}{2\pi}\int_{-\pi}^{\pi}\left(\sum_{k=-\infty}^{\infty}\frac{\Gamma(1+|k|-\beta)}{\Gamma(1-\beta)\Gamma(1+|k|)}r^{|k|}e^{ik(x-t)}\right)g(t)dt,
\end{equation}
where 
\begin{equation}
g(t)=\sum_{k=-\infty}^{\infty}\frac{\Gamma(1-\beta)\Gamma(1+|k|)}{\Gamma(1+|k|-\beta)}e^{ixk}.
\end{equation}
\end{proof}
\begin{lemma}
 Let $f(x)\in B_2^{\alpha},$ where $0<\alpha<1$ and $0<\beta\leq \alpha$. Then there is a subset 
$F\in\{z; |z|=1\}$ satisfying condition 
\begin{equation}
C_{\alpha-\beta}(F)=0
\end{equation}
end for each $e^{i x}\in F$ there is a number $A(x)$ such that the following inequality holds
\begin{equation}
|F(z_1)-F(z_2)|<A(x)|z_1-z_2|^{\alpha-\beta},
\end{equation}
where
\begin{equation}
|e^{ix}-z_j|<2(1-|z_j|),\quad j=1,2,
\end{equation}

	Here $F(z)$ is the harmonic function with the boundary values $f(x)$, i.e.
\begin{equation}
\lim_{r\to 1-}F(re^{ix})=f(x).
\end{equation}
\end{lemma}
\begin{proof}
 Thanks of the previous lemma there is a function $g(x)\in B_2^{\alpha-\beta}$ such that
\begin{equation}
\frac{\partial F(z)}{\partial z}=\frac{1-\beta}{2\pi}\int_{-\pi}^{\pi}\frac{g(t)}{(z-e^{it})^{2-\beta}}dt.
\end{equation}
Taking into account the inequalities 
\begin{equation}
1-|z|<|e^{it}-z|,\quad \left|e^{it}-\frac{z}{|z|}\right|\leq 2|e^{it}-z|,
\end{equation}
valid for each $|z|<1,$ we have 
\begin{equation}
\frac{(1-|z|)^{1-\beta}}{1-\beta}\left|\frac{\partial F(z)}{\partial z}\right|\leq
\end{equation}
\begin{equation}
\leq \frac{1}{2\pi}\int_{-\pi}^{\pi}\frac{(1-|z|)^{1-\beta}}{|z-e^{it}|^{2-\beta}}|g(t)|dt\leq
\end{equation}
\begin{equation}
\leq \frac{1}{2\pi}\int_{|e^{it}-z/|z||\leq y}\frac{|g(t)|}{(1-|z|)}dt+\frac{1}{2\pi}\int_{|e^{it}-z/|z||>y}\frac{(1-|z|)^{1-\beta}}{|e^{it}-z/|z||^{2-\beta}}|g(t)|dt.
\end{equation}
If $z,\quad |e^{ix}-z|\leq 2(1-|z|)$, then putting $y=(1-|z|)$ we get
\begin{equation}
(1-|z|)^{1-\beta}\left|\frac{\partial F(z)}{\partial z}\right|\leq
\end{equation}
\begin{equation}
\leq \frac{1}{2\pi}\int_{|t-x|\leq (1-|z|)}\frac{|g(t)|}{(1-|z|)}dt+(1-|z|)^{1-\beta}\int_{1-|z|}^2\frac{1}{y^{2-\beta}}d\left(\int_{x-y}^{x+y}|g(t)|dt\right)\leq
\end{equation}
\begin{equation}
\leq g(t)^*+(1-|z|)^{1-\beta}\int_{1-|z|}^2\frac{1}{y^{3-\beta}}\left(\int_{x-y}^{x+y}|g(t)|dt\right)dy.
\end{equation}
Consequently,
\begin{equation}
(1-|z|)^{1-\beta}\left|\frac{\partial F(z)}{\partial z}\right|\leq g^*(x).
\end{equation}
	Let $l$ be the linear interval with the end points $z_1, z_2\in\{z,\quad |e^{ix}-z|\leq 2(1-|z|)\}$. 
We have
\begin{equation}
|F(z_1)-F(z_2)|\leq \left|\int_{l}\frac{\partial F(z)}{\partial z}dz\right|\leq
\end{equation}
\begin{equation}
\leq g^*(x)\int_{l}\frac{|dz|}{(1-|z|)^{1-\alpha+\beta}}\leq g^*(x)(|z_1-z_2|^{\alpha-\beta}.
\end{equation}

	To complete the prove it is sufficient to note, that 
\[
C_{\alpha-\beta}(F)=0,
\]
where $F=\{x;\quad g^*(x)=\infty\}$.
\end{proof}
	The proof of the following lemma may be found in \cite{7}.
\begin{lemma} 
 Let $g(t),\quad 0<t,$ be a positive and no decreasing function. Let $u(z)$ be a non-negative 
harmonic function defined on the unit disk. Then, the subset
\begin{equation}
F=\left \{ \xi; |\xi|=1,\quad 
\sup_{z\in \Delta(\xi)}\frac{u(z)}{g(1-\vert z\vert)}=\infty\right\}
\end{equation}
has zero Hausdorff's measure, i.e.
\begin{equation}
M_{h(t)}(F)=0,
\end{equation}
where $h(t)=tg(t).$
\end{lemma}
\section{Uniqueness theorem for analytic functions}

	The main result of this section is the following theorem.
\begin{theorem}
 Let $\{\xi_n\}$ be a sequence in the unit disk with
\begin{equation}
\lim_{n\to \infty}\vert \xi_n \vert=1.
\end{equation}
Let $E$ be a subset of unit circle such that for some continuous function $0\leq h(t),\quad 0<t,$ 
satisfying the condition 
\begin{equation}
\lim_{t\to 0+}\frac{h(t)}{t\log \frac{1}{t}}=0
\end{equation}
we have $M_{h}(E)>0$.

	Let for each point $y\in E$ there is a subsequence $\{\xi_{n_k}\}$ such that
\begin{equation}
\left\vert y-\frac{\xi_{n_k}}{\xi_{n_k}} \right\vert<2(1-\vert \xi_{n_k}\vert), \quad k=1,2,...
\end{equation}
	Let $0\leq\alpha<1$ be a fixed number and $f(z)$ is an analytic function with
\begin{equation}
\int_0^1\int_{-\pi}^{\pi}|f'(z)|^2(1-|z|)^{\alpha}dxdy<\infty,
\end{equation}
and
\begin{equation}
\lim_{n\to\infty}f(\xi_n)=0 \quad (1)
\end{equation}
then $f(z)\equiv 0$.
\end{theorem}
\begin{proof}
 At first let us note, that instead of (1) we can assume
\begin{equation}
\lim_{n\to\infty}\frac{f(\xi_n)}{(1-|\xi_n|)^{\alpha-\beta}}=0,
\end{equation}
where $0<\beta<\alpha$ is a constant. Indeed, By lemma 3 there is a subset $F_1$, for which
\begin{equation}
C_{\alpha-\beta}(F_1)=0,
\end{equation}
end for each point $x\notin F_1$ there is a number $A(x)<\infty$ such that for arbitrary points $z_1,z_2$ 
from the unite disk, satisfying the condition
\begin{equation}
|e^{ix}-z_j|<2(1-|z_j|),\quad j=1,2,
\end{equation}
we have
\begin{equation}
|f(z_1)-f(z_2)|<A(x)|z_1-z_2|^{\alpha-\beta}.
\end{equation}
In particularly, if $x\in E\setminus F_1$ then
\begin{equation}
|f(\xi_n)|<A(x)(1-|\xi_n|)^{\alpha-\beta}
\end{equation}
for each point $\xi_n$ satisfying the condition
\begin{equation}
|e^{ix}-\xi_n|<2(1-|\xi_n|).
\end{equation}
Thanks to lemma 1 we have 
\begin{equation}
M_{h}(F_1)=0
\end{equation}
and so, 
\begin{equation}
M_{h}(E\setminus F_1)>0.
\end{equation}

	For each point $z,\quad |z|<1,$ let us denote
\begin{equation}
C(z)=\{w;\quad |w|=1,\quad |w-z|\leq 2(1-|z|)\}.
\end{equation}

	For the zeros $\Lambda=\{z_k\}_{k=1}^{\infty}$ of our function $f(z)$ we have
\begin{equation}
\sum_{k=1}^{\infty}(1-|z_k|)<\infty
\end{equation}
and, see \cite{6}, we have also 
\begin{equation}
\frac{1}{2\pi}\int_{-\pi}^{\pi}\log\rho_{\frac{1+\alpha}{2}}(e^{ix},\Lambda)dx>-\infty,
\end{equation}
where 
\begin{equation}
\rho_{\sigma}(\xi,\Lambda)=\inf_{z\in \Lambda}\frac{|\xi-z|}{(1-|z|)^{\sigma}}.
\end{equation}
	Let us consider a new function
\begin{equation}
G(w)=\sum_{k=1}^{\infty}\chi_k(w),\quad w\in \partial D,
\end{equation}
where $\chi_k(w)$ is the characteristic function of the arc $C(z_k)$.

	We want to prove that the subset
\begin{equation}
F_2=\{w;\quad w\in \partial D,\quad G(w)=+\infty\}
\end{equation}
satisfies the condition
\begin{equation}
M_{t\log1/t}(F_2)=0.		\quad(2)
\end{equation}

	Let us suppose $M_{t\log1/t}(F_2)>0.$ Then, see \cite{MR0225986}, p. 18, there is a compact subset 
$F\subset F_2$ for which 
\begin{equation}
M_{t\log 1/t}(F)>0.
\end{equation}
For each natural $N$, the family of subsets $C(z_k),\quad k=N, N+1, \dots$ cover $F$. 
By Alfor's theorem, see \cite{MR0225986}, from that family of arcs we can choose a finite number, 
which cover $F$ and have a finite multiplicity less an absolute constant $A$. Let 
\begin{equation}
C(z_{k_1}),\dots, C(z_{k_{m}})
\end{equation}
be the constructed subfamily, which cover the set $F$. 

	We can write
\begin{equation}
\left(1-\frac{1+\alpha}{2\sigma}\right)\sum_{j=1}^{m}|C(z_{n_j})|\log \frac{e}{|C(z_{n_j})|}\leq 
\end{equation}
\begin{equation}
\leq \left(1-\frac{1+\alpha}{2\sigma}\right)\sum_{j=1}^{m_1}\int_{C(z_{n_j})}\log^+\left(\frac{1}{|w-z_{n_j}|}\right)|dw|\leq
\end{equation}
\begin{equation}
\leq \sum_{j=1}^{m_1}\int_{C(z_{n_j})}\log^+\left(\frac{(1-|z_{n_j}|)^{(1+\alpha)/2}}{|w-z_{n_j}|}\right)|dw|\leq
\end{equation}
\begin{equation}
\leq A\int_{Q_N}\log^+\left(\sup_k\frac{(1-|z_{k}|)^{(1+\alpha)/2}}{|w-z_{k}|}\right)|dw|,
\end{equation}
where 
\begin{equation}
Q_N=\bigcup_{j=N}^{\infty}C(z_j).
\end{equation}
Letting $N$ to go infinity we get 
\begin{equation}
M_{t\log1/t}(F)=0.
\end{equation}
The received contradiction proves (2). Consequently,
\begin{equation}
M_{h}(F_2)=0.
\end{equation}
	Thus, for each point $e^{ix}\in E\setminus F_2$ in domain
\begin{equation}
\{z;\quad |e^{ix}-z|<2(1-|z|)\}
\end{equation}
there are only finite number zeros of the function $f(z)$.

	By F. Riesz theorem we have the representation
\[
f(z)=B(z)F(z),
\]
where $B(z)$ is the Blashke product constructed by zeros $\{w_n\}$ 
of the function $f(z)$ and $F(z)\in H^{\infty}$, which has no zeros in the unit disk. 

	Let us denote 
\begin{equation}
v(z)=\sum_{n=1}^{\infty}\frac{1-|z|^2}{\left|z-\frac{w_n}{|w_n|}\right|}(1-|w_n|),\quad |z|<1.
\end{equation}
For two arbitrary points $z,\quad w$ from the unit disk we have 
\begin{equation}
-\frac{(1-|z|^2)(1-|w|^2)}{|z-w|^2}\leq -\log\left(1+\frac{(1-|z|^2)(1-|w|^2)}{|z-w|^2}\right)=
\log\left|\frac{w-z}{1-\bar{w}z}\right|^2.
\end{equation}
Let $y$ be a point on the unit circle and $n_0$ be a natural number. 
Let for each index $n=n_0, n_0+1, \dots$ the inequalities
\begin{equation}
\left|y-\frac{w_n}{|w_n|}\right|\geq 4(1-|w_n|),\quad n=1, 2, \dots
\end{equation}
hold. Then for arbitrary point $z,\quad |z|<1,$ satisfying the condition 
\begin{equation}
\left|y-\frac{z}{|z|}\right|\leq 2(1-|z|).
\end{equation}
there is a constant $C>0$ such that
\begin{equation}
-Cv(z)\leq \log |B(z)|.
\end{equation}
	Thanks to the lemma 4, applied to the function $v(z)\geq 0$, 
we get that there is a subset $F_3\in \partial D$ for which 
\begin{equation}
M_{h(t)}(F_3)=0,
\end{equation}
where $h(t)=tg(t)$ and for each point $y\in \partial D\setminus F_3$ we have
\begin{equation}
\sup\left\{\frac{v(z)}{g(1-|z|)};|z|<1,\quad \left|y-\frac{z}{|z|}\right|\leq 2(1-|z|)\right\}<+\infty.
\end{equation}

	Consequently, there is a subset $F_4$ such that 
\begin{equation}
M_{h(t)}(F_4)=0
\end{equation}
and for each point $y\notin E\setminus F_4$ we have
\begin{equation}
\sup\left\{\frac{|f(z)|}{(1-|z|)^{1-\alpha}};|z|<1,\quad \left|y-\frac{z}{|z|}\right|\leq 2(1-|z|)\right\}<\infty.
\end{equation}
	These remarks contradict theorem's conditions since 
$$E\setminus (F_1\cup F_2\cup F_3\cup F_4)\neq \emptyset.$$
\end{proof}

\section{De-noising problem for analytic functions}

    In this section we consider the following problem: let $f(z)$ be a bounded analytic function
and $\{z_n\}_{n=1}^{\infty}$ be a sequence from the unit disk. Let we can calculate empirically 
the values of this function at the points $\{z_n\}_{n=1}^{\infty}$ with some error, i.e.
\begin{equation}
w_n=f(z_n)+\epsilon_n,\quad n=1,\dots, 
\end{equation}
where $\epsilon_n,\quad n=1,\dots$ is a sequence of independent random quantities with the same 
distribution and with the mean value equal zero. The following problem naturally arises: is it 
possible to choose the points $\{z_n\}_{n=1}^{\infty}$ in such a way that by observed quantities 
$w_n,\quad n=1, 2, \dots$ it will be possible to restore the function $f(z)$ by probability one? 

	The relation of this problem with the identification problem for linear bounded systems one can 
find in \cite{9}.

	Here we give some classical results of Shizuo Kakutani, see \cite{MR0023331}, which play a 
principal role in answer to this question.

	Let $\Omega$ be an arbitrary set and let $\sigma$ be a $\sigma-$field of subsets of 
$\Omega$. Let $\Re(\sigma)$ be the
family of all countable additive measures $\mu(d\omega)$ defined on $\sigma$ for which 
$\mu(\Omega)=1$. 
\begin{definition}
Two measures $\mu,\nu \in \Re(\sigma)$  called orthogonal (notation $\mu\perp \nu$) if 
there are disjoint subsets $B,B'\in \sigma$ such that 
\begin{equation}
\mu(B)=\nu(B')=1.
\end{equation}
\end{definition}
    Let $\mu, \nu \in \Re(\sigma)$ be measures on $(\Omega,\sigma)$. For arbitrary measure
$\tau\in \Re(\sigma)$ such that $\mu$ and $\nu$ are absolutely continuous in respect to $\tau$, 
let us denote
\begin{equation}
\rho(\mu,\nu)=\int_{\Omega}\sqrt{\frac{\mu(d\omega)}{\tau(d\omega)}}\sqrt{\frac{\nu(d\omega)}{\tau(d\omega)}}\tau(d\omega).
\end{equation}
This integral doesn't depend upon the choice of the measure $\tau$. That is way the following 
E.Hellinger's notation 
\begin{equation}
\rho(\mu,\nu)=\int_{\Omega}\sqrt{\mu(d\omega) \nu(d\omega)}
\end{equation}
is natural.

    Let $\{\mu_n\}_{n=1}^{\infty}$ and $\{\nu_n\}_{n=1}^{\infty}$ be two family of
probability measures on ${\bf C}$. Let us denote by $\mu=\mu_1\times\mu_2\times\dots,\quad \nu=\nu_1\times\nu_2\times\dots.$ the infinite direct products. 

    It is easy to see that if for some $k_0$ we have $\mu_{k_0}\bot\nu_{k_0}$ then $\mu\bot\nu$.
The case $\mu_k\sim\nu_k,\quad k=1,\dots$ was considered in \cite{MR0023331}.
\begin{theorem}
Let $\mu_k\sim\nu_k$ for all $k=1,\dots$. Then the measures $\mu$
and $\nu$ are equivalent if and only if $$ \prod_{k=1}^{\infty}
\rho(\mu_k,\nu_k)>0. $$ Otherwise those measures are orthogonal,
i.e. $\mu\perp\nu$. 
\end{theorem} 
    Here, we need only the following particular case of S. Kakutani's Theorem.
    Let $f(z),g(z)\quad |z|<1,$ be bounded analytic functions and $z_k,\quad k=1,2,\dots$ be points 
in the unite disk. Let 
\begin{equation}
d\mu_k(z)=P(z-f(z_k))dxdy,\quad z=x+iy
\end{equation}
and 
\begin{equation}
d\nu_k(z)=P(z-g(z_k))dxdy,
\end{equation} 
where $P(z)\geq 0$ and
\begin{equation}
\int_{-\infty}^{\infty}\int_{-\infty}^{\infty}P(z)dxdy=1.
\end{equation}
	We have
\begin{equation}
\rho(d\mu_k,d\nu_k)=\int_{-\infty}^{\infty}\int_{-\infty}^{\infty}\sqrt{P(z-f(z_k))P(z-g(z_k))}dxdy.
\end{equation}
and
\begin{equation}
1-\rho(d\mu_k,d\nu_k)=
\end{equation}
\begin{equation}
=\frac{1}{2}\int_{-\infty}^{\infty}\int_{-\infty}^{\infty}\left(\sqrt{P(z-f(z_k))}-\sqrt{P(z-g(z_k))}\right)^2dxdy\geq A|f(z_k)-g(z_k)|^2
\end{equation}
for some number $A>0$.

	The corresponding infinite products of measures are orthogonal if 
\begin{equation}
\sum_{k=1}^{\infty}|f(z_k)-g(z_k)|^2=\infty.
\end{equation}
	So, if we have
\begin{equation}
w_n=H(z_n)+\epsilon_n,\quad n=1,\dots,
\end{equation}
then by probability one, it is possible to identify $H(z)$ with some $f(z)$, if
the points $\{z_n\}_{n=1}^{\infty}$ are possible to choose in such
a way, that from the condition $f(z),g(z)\in H^{\infty}$ and
\begin{equation}
\sum_{k=1}^{\infty}|f(z_k)-g(z_k)|^2<\infty
\end{equation} 
it follows $f(z)\equiv g(z)$. 
This note permits us to formulate the following result.
\begin{theorem}
 Let $0\leq\alpha<1$ be a fixed number. Let $\{z_n\}$ be a sequence in the unit disk with 
\begin{equation}
\lim_{n\to \infty}\vert z_n \vert=1.
\end{equation}
Let $E$ be a subset of unit circle such that for some continuous function $0\leq h(t),\quad 0<t,$ 
satisfying the condition 
\begin{equation}
\lim_{t\to 0+}\frac{h(t)}{t\log\frac{1}{t}}=0
\end{equation}
we have 
\begin{equation}
M_{h}(E)>0.
\end{equation}

	Let for each point $y\in E$ there is a subsequence $\{z_{n_k}\}$ such that
\begin{equation}
\left\vert y-\frac{z_{n_k}}{z_{n_k}} \right\vert<2(1-\vert z_{n_k}\vert), \quad k=1,2,...
\end{equation}
	Let we can observe the quantities
\begin{equation}
X_n=S(z_n)+\xi_n,\quad n=1, 2, \dots
\end{equation}
where $\xi_n,\quad n=1, 2, \dots$ are independent random variables with the same absolutely continuous 
distributions and zero mean values. Let $S(z)$ be a bounded analytic function and
\begin{equation}
\int_0^1\int_{-\pi}^{\pi}|S'(z)|^2(1-|z|)^{\alpha}dxdy<\infty,\quad 0<\alpha<1.
\end{equation}
	Then, by $X_n,\quad n=1,2,\dots$ it is possible to restore the 
function $S(z)$ by probability one.
\end{theorem}

\end{document}